\documentclass{article}
\usepackage{calrsfs}
 \usepackage{amsmath}
\usepackage{graphicx}
\usepackage[utf8]{inputenc}
\usepackage[english]{babel}
 \usepackage{amssymb}
\usepackage{pst-node}
\usepackage{authblk}
\usepackage{tikz}
\usepackage{tikz-cd}
\usepackage{amsthm}
\usepackage{graphicx}
\usepackage{url}
\usepackage{hyperref}
\usepackage[utf8]{inputenc}
\usepackage[all,cmtip]{xy}

\theoremstyle{definition}
\DeclareMathOperator*{\res}{Res}

\newtheorem{definition}{Definition}[section]
\newtheorem{remark}{Remark}[section]

\newtheorem{theorem}{Theorem}[section]
\newtheorem{corollary}{Corollary}[theorem]
\newtheorem{lemma}[theorem]{Lemma}
\newtheorem{proposition}[theorem]{Proposition}
\begin{document}
 
\title{Open Hurwitz Flat F manifolds}
\author{Guilherme F. Almeida}
\affil{}
\date{}

\maketitle

\begin{abstract}

In this paper, we derive  Open WDVV equations starting from any Hurwitz Dubrovin Frobenius manifold. The WDVV equations plays a crucial role in the structure of Frobenius manifolds, quantum cohomology, and integrable systems. Extending these ideas, Open WDVV equations provide a framework to incorporate boundary conditions, making them fundamental in Open Gromov-Witten theory. Using Dubrovin’s construction of Landau-Ginzburg superpotentials associated with Hurwitz spaces, we demonstrate that their primitives satisfy  Open WDVV equations. Our approach provides an efficient method for computing Open WDVV equations associated with any Hurwitz Dubrovin Frobenius manifold.

\end{abstract}

\tableofcontents

\section{Introduction}
The main goal of this paper is to derive an Open WDVV equation starting from any Hurwitz Dubrovin Frobenius manifold. The study of WDVV equations plays a fundamental role in the mathematical formulation of topological field theories. These equations arise naturally in the study of Frobenius manifolds, quantum cohomology, and integrable systems. In particular, the classical WDVV equation governs the structure of the quantum cohomology ring, encoding the associativity of the quantum product and providing deep connections with singularity theory and the geometry of moduli spaces of curves. However, many physical and geometric problems require a generalization of the classical WDVV equations to open settings, where boundary conditions play an essential role. The Open WDVV equation extends the classical equation by introducing an additional structure that accounts for open string invariants, making it an important tool in Open Gromov-Witten theory. \\

The classical WDVV equations governs the structure of a potential function associated with a Frobenius manifold. The following definition introduces the fundamental conditions that a function must satisfy to be a solution to this equation.

\begin{definition}\cite{B. Dubrovin2}
An analytic function $F(t)$, where $t = (t^1, t^2, \ldots, t^n) \in U \subset \mathbb{C}^n$ defined in an open subset of $\mathbb{C}^n$, is considered a solution of the WDVV (Witten-Dijkgraaf-Verlinde-Verlinde) equations if its third derivatives
\begin{equation*}
c_{\alpha\beta\gamma} = \frac{\partial^3 F}{\partial t^{\alpha} \partial t^{\beta} \partial t^{\gamma}}
\end{equation*}
satisfy the following conditions:

\begin{enumerate}

\item The coefficients $\eta_{\alpha\beta} = c_{1\alpha\beta}$ form elements of a constant nondegenerate matrix.
\item The quantities $c_{\alpha\beta}^{\gamma} = \eta^{\gamma\delta} c_{\alpha\beta\delta}$ represent the structure constants of an associative algebra. i.e.,
\begin{equation}\label{closed WDVV equation}
\frac{\partial^3 F}{\partial t^{\alpha}\partial t^{\beta}\partial t^{\mu}}\eta^{\mu\lambda}\frac{\partial^2 F^o}{\partial t^{\lambda}\partial t^{\gamma}}=\frac{\partial^3 F}{\partial t^{\gamma}\partial t^{\beta}\partial t^{\mu}}\eta^{\mu\lambda}\frac{\partial^2 F^o}{\partial t^{\lambda}\partial t^{\alpha}}
\end{equation}

\item The function $F(t)$ must be quasi-homogeneous.
\end{enumerate}

\end{definition}

A natural extension of the closed WDVV equations is given by the Oriented WDVV equations, which involves a vector-valued function rather than a single potential. This formulation introduces additional flexibility and leads naturally to the Open WDVV equations.

\begin{definition}\cite{Alexander}
The analytic vector valued functions $(F_1(t), F_2(t), ..., F_n(t))$ defined on an open subset $U \subset \mathbb{C}^n$, where $t = (t_1, t_2, ..., t_n)$, is a solution of an Oriented WDVV equations if the second derivatives 
\begin{equation*}
c_{\alpha\beta}^{\gamma} = \frac{\partial^2F^{\gamma}}{\partial t^{\alpha} \partial t^{\beta}}
\end{equation*}
 satisfy the following conditions:

\begin{enumerate}
    \item The $c_{\alpha \beta}^{\gamma}$ are structure constants of an associative algebra, expressed as:

 \begin{equation}\label{Oriented WDVV equation}
    \frac{\partial^2F^\alpha}{\partial t^\beta \partial t^\mu} \frac{\partial^2F^\mu}{\partial t^\gamma \partial t^\delta} =
    \frac{\partial^2F^\gamma}{\partial t^\mu \partial t^\delta} \frac{\partial^2F^\beta}{\partial t^\gamma \partial t^\mu}, \quad 1 \leq \alpha, \beta, \gamma, \delta \leq n. 
\end{equation}

    \item The vector $(F^1(t), F^2(t), ..., F^n(t))$ must be a quasi-homogeneous function.
\end{enumerate}

Additionally, a closely related concept to Oriented WDVV equations is the Open WDVV equations, providing an intriguing extension of the closed WDVV equations (\ref{closed WDVV equation}).
\end{definition}

\begin{definition}\cite{Alexander}
For a fixed solution $F(t^1, t^2, ..., t^n)$ of the closed WDVV equations (\ref{closed WDVV equation}), the Open WDVV equations form the following system of partial differential equations for an analytic function $F^o(t_1, t_2, ..., t_n, t_{n+1})$ defined in some open subset of $\mathbb{C}^{n+1}$ if 

\[
c_{\alpha \beta}^{\gamma}=
\begin{cases}
    \eta^{\gamma\delta}\frac{\partial^3 F}{\partial t^{\alpha}\partial t^{\beta}\partial t^{\gamma}} & \text{if } 0 \leq\gamma\leq n, \\
   \frac{\partial^2 F^o}{\partial t^{\alpha}\partial t^{\beta}}  & \text{if } \gamma=n+1.
\end{cases}
\] satisfying the following conditions:

\begin{enumerate}
    \item The $c_{\alpha \beta}^{\gamma}$ are structure constants of an associative algebra, expressed as:
\begin{equation}\label{Open WDVV equation 0}
\begin{split}
\frac{\partial^3 F}{\partial t^{\alpha}\partial t^{\beta}\partial t^{\mu}}\eta^{\mu\lambda}\frac{\partial^2 F^o}{\partial t^{\lambda}\partial t^{\gamma}}+\frac{\partial^2 F^o}{\partial t^{\alpha}\partial t^{\beta}}\frac{\partial^2 F^o}{\partial t^{n+1}\partial t^{\gamma}}=\frac{\partial^3 F}{\partial t^{\gamma}\partial t^{\beta}\partial t^{\mu}}\eta^{\mu\lambda}\frac{\partial^2 F^o}{\partial t^{\lambda}\partial t^{\alpha}}+\frac{\partial^2 F^o}{\partial t^{\gamma}\partial t^{\beta}}\frac{\partial^2 F^o}{\partial t^{n+1}\partial t^{\alpha}}
\end{split}
\end{equation}
where, $1 \leq \alpha, \beta, \gamma, \delta \leq n$.
    
    \item The vector $(F^1(t), F^2(t), ..., F^n(t))$ must be a quasi-homogeneous function.
\end{enumerate}
\end{definition}

In \cite{Alexander}, Paolo Rossi made the following remark

\begin{lemma} \cite{Alexander}
Consider a pair $(F, F^o)$ as a solution to  Open WDVV equations (\ref{Open WDVV equation}). Then, the vector-valued functions 
\begin{equation*}
\left(\eta^{1\mu} \frac{\partial F}{\partial t_\mu}, \eta^{2\mu} \frac{\partial F}{\partial t_\mu}, \ldots,\eta^{n\mu} \frac{\partial F}{\partial t_\mu}, F^o\right)
\end{equation*}
form a solution to  Oriented WDVV equations (\ref{Oriented WDVV equation}).

\end{lemma}

In \cite{Alcolado Adam},  Alcolado derived an equivalent equation for an Open WDVV equation.

\begin{theorem}\cite{Alcolado Adam}\label{Adam identity 10}
For a fixed solution $F(t^1, t^2, ..., t^n)$ of the closed WDVV equations (\ref{closed WDVV equation}). Let $\lambda= \lambda(p,t^1, . . . , t^n)$ be a function satisfying $\frac{\partial \lambda}{\partial p}\neq  0$ and
 \begin{equation}\label{Adam identity 1 equation1}
\partial_{\alpha}\partial_{\beta}\left(\int \lambda(p)dp\right)=\frac{ \partial_{\alpha}\lambda \partial_{\beta}\lambda-c_{\alpha\beta}^{\gamma}\partial_{\gamma}\lambda}{ \partial_{p}\lambda},
\end{equation}
where $c_{\alpha\beta}^{\gamma}=\frac{\partial^3 F}{\partial t^{\alpha}\partial t^{\beta}\partial t^{\gamma}}$.
Then the pair $(F, \int \lambda(p) dp)$ do satisfy  Open WDVV equations (\ref{Open WDVV equation 0}) with respect the closed WDVV solution $F(t^1, t^2, ..., t^n)$.

\end{theorem}

A remarkable application of the theorem \ref{Adam identity 10} is that the primitive of the Landau Ginzburg superpotential of $A_n$ singularities do satisfy (\ref{Adam identity 1 equation1}). In this way, we do have an efficient way to compute the Open WDVV equations (\ref{Open WDVV equation 0})  for any $A_n$ Dubrovin Frobenius potential.\\

\begin{theorem}\cite{Alcolado Adam}\label{Adam identity 2}

Let the  $F(t^1, t^2, ..., t^n)$ and
\begin{equation}\label{superpotential An}
\lambda(p)=\prod_{i=0}^n (p-x_i), \quad \sum_{i=0}^n x_i=0
\end{equation}
be the closed WDVV solution and the superpotential associated  with $A_n$ Dubrovin Frobenius manifold respectively . Then, the primitive of superpotential (\ref{superpotential An}) does satisfy (\ref{Adam identity 1 equation1}) with respect the WDVV associated with $A_n$ WDVV solution. Moreover, the pair
\begin{equation*}
\left(F, \int \lambda(p) dp\right)
\end{equation*}
do satisfy Open WDVV equations (\ref{Open WDVV equation 0}) with respect the closed $A_n$ WDVV solution $F(t^1, t^2, ..., t^n)$.
\end{theorem}

In chapter 5 of  \cite{B. Dubrovin2}, Dubrovin constructed a Landau-Ginzburg superpotential associated to any Hurwitz space, which are moduli space of meromorphic covering over $\mathbb{CP}^1$ with fixed ramification profile. In particular, the  Landau-Ginzburg superpotential associated with the $A_n$ Dubrovin Frobenius manifold is the Hurwitz space $H_{0,n}$, which moduli space of genus 0 covering over  $\mathbb{CP}^1$ with a pole of order $n+1$ over $\infty \in \mathbb{CP}^1$. 

Using the Hurwitz Dubrovin Frobenius manifold construction, we generalize Theorem \ref{Adam identity 2} to any Hurwitz space, proving the main result of this paper.

\begin{theorem}\label{Main theorem Almeida}

Let the  $F(t^1, t^2, ..., t^n)$ and $\lambda(p)$
be the closed WDVV solution and the superpotential associated  with any Hurwitz Dubrovin Frobenius manifold respectively . Then, the primitive of superpotential $\lambda(p)$ does satisfy (\ref{Adam identity 10}) with respect the WDVV solution $F(t^1, t^2, ..., t^n)$. Moreover, the pair
\begin{equation*}
\left(F, \int \lambda dp\right)
\end{equation*}
do satisfy  Open WDVV equations (\ref{Open WDVV equation 0}) with respect the closed  WDVV solution $F(t^1, t^2, ..., t^n)$.
\end{theorem}

This result extends Alcolado’s approach beyond $A_n$ singularities, providing a broad framework for computing Open WDVV equations in the context of Hurwitz spaces.

\subsection*{Acknowledgements}
I am grateful to Professor Walcher for pointing out Alcolado's formula for WDVV extensions and for introducing me to the Open Gromov–Witten framework. Furthermore, I would also like to thank Professor Hertling for fruitful discussions regarding F-manifolds.

\section{Flat F manifold and Dubrovin Frobenius manifold}\label{Flat F manifold and Dubrovin Frobenius manifold}

The main goal of this section is to give a brief overview of the geometric aspects of closed and oriented WDVV equations.\\

In chapter 1 of \cite{B. Dubrovin2}, Dubrovin formulated a geometric interpretation of the WDVV equations which is given by the following

\begin{definition}
A Frobenius Algebra $\mathcal{A}$ is a unital, commutative, associative algebra equipped with an invariant, non-degenerate bilinear pairing 
\begin{equation*}
\eta : \mathcal{A} \otimes \mathcal{A} \mapsto \mathbb{C},
\end{equation*}
which is invariant in the following sense:
\begin{equation*}
\eta(A \bullet B, C) = \eta(A, B \bullet C), \quad \forall A, B, C \in \mathcal{A}.
\end{equation*}
\end{definition}

\begin{definition} \cite{B. Dubrovin2},\cite{B. Dubrovin3}
Let $M$ be a complex manifold of dimension $n$. A Dubrovin Frobenius structure over $M$ consists of the following compatible objects:
\begin{enumerate}
\item A family of Frobenius multiplications $\bullet_p : T_pM \times T_pM \mapsto T_pM$ that are analytically dependent on $p \in M$. This family induces a Frobenius multiplication on 
\begin{equation*}
\bullet : \Gamma(TM) \times \Gamma(TM) \mapsto \Gamma(TM).
\end{equation*}
\item A flat pseudo-Riemannian metric $\eta$ on $\Gamma(TM)$, also known as the Saito metric.
\item A unity vector field $e$ that is covariantly constant with respect to the Levi-Civita connection $\nabla$ for the metric $\eta$, i.e., $\nabla e = 0$.
\item Consider the tensor $c(X, Y, Z) := \eta(X \bullet Y, Z)$. We require the 4-tensor
\begin{equation*}
(\nabla_W c)(X, Y, Z)
\end{equation*}
to be symmetric with respect to $X, Y, Z, W \in \Gamma(TM)$.
\item An Euler vector field $E$ determined by:
\begin{equation*}
\nabla \nabla E = 0, \quad \mathcal{L}_E \eta(X, Y) = (2 - d) \eta(X, Y), \quad \mathcal{L}_E c(X, Y, Z) = c(X, Y, Z),
\end{equation*}
where $X, Y, Z \in \Gamma(TM)$. Moreover, we require $\nabla E$ to be diagonalizable.

Let $(t^1, t^2, \ldots, t^n)$ be the flat coordinates with respect to the metric $\eta$. These coordinates are denoted as Saito flat coordinates. The Euler vector $E$ can be explicitly represented as:
\begin{equation*}
E = \sum_{i=1}^n ((1 - q_i) t_i + r_i) \partial_i.
\end{equation*}
\end{enumerate}
\end{definition}

Roughly speaking, Dubrovin Frobenius manifold is the geometric structure that naturally arise in the domain of any WDDV solution, which is given by a family of Frobenius algebra on the sheaf of holomorphic vector fields, a flat structure and some suitable marked vector fields. An important example of WDVV solutions are the generating function of Gromov Witten invariants called Gromov Witten potential. Another source of Dubrovin Frobenius manifolds comes from Landau-Ginzburg superpotential, which are  unfolding of singularities or family of covering over $\mathbb{CP}^1$.\\

At this point, we recall the definition of an F-manifold, which is a weaker version of a Dubrovin Frobenius manifold, as follows.\\

\begin{definition}\cite{David} \cite{Hertling}
An F-manifold $(M, \bullet, e)$ is a complex manifold M which have
\begin{enumerate}
\item A family of commutative and associativity $\bullet_p : T_pM \times T_pM \mapsto T_pM$ that are analytically dependent on $p \in M$. This family induces a multiplication on 
\begin{equation*}
\bullet : \Gamma(TM) \times \Gamma(TM) \mapsto \Gamma(TM).
\end{equation*}

\item A unity vector field $e$.
\item The following integrability condition holds
\begin{equation*}
Lie_{X\bullet Y}(\bullet)=X\bullet Lie_{Y}+Y\bullet Lie_{X}, \quad X,Y \in \Gamma(TM).
\end{equation*}
\item An F-manifold $(M,\bullet,e,E)$ with a Euler vector field is an F-manifold $(M, \bullet, e)$ together with a vector field E called Euler vector field such that
\begin{equation*}
Lie_E(\bullet)=\bullet.
\end{equation*}
\end{enumerate}
\end{definition}

Following the point of view of \cite{David}, \cite{Hertling} , we will define an enrichments of an
F-manifold until we construct a Dubrovin Frobenius manifold.\\

\begin{definition}\cite{David} \cite{Hertling}
\begin{enumerate}
\item A flat F manifold $(M,\bullet,e,\nabla)$ is an F manifold $(M,\bullet,e,)$ together with a flat, torsion free connection $\nabla$ on the $\Gamma(TM)$ which satisfy
\begin{equation}\label{Hertling integrability condition}
\begin{split}
\nabla_X(Y\bullet )-\nabla_Y(X\bullet )&=[X,Y]\bullet, \quad X,Y \in \Gamma(TM),\\
\end{split}
\end{equation}
\begin{equation*}
\begin{split}
&\nabla_Xe=0, \quad X \in \Gamma(TM).\\
\end{split}
\end{equation*}

\item A flat F manifold $(M,\bullet,e,\nabla,E)$ with an Euler vector is a flat F manifold $(M,\bullet,e,\nabla)$ together an Euler vector field E such that
\begin{equation*}
\begin{split}
&\nabla_X\nabla_YE=0, \quad X,Y \in \Gamma(TM).\\
\end{split}
\end{equation*}

\item A Dubrovin Frobenius manifold $(M,\bullet,e,\eta)$ (without an Euler vector field) is an F manifold $(M,\bullet,e)$ together with a holomorphic bilinear paring $\eta$ on the $\Gamma(TM)$ which is multiplication invariant. i.e.
\begin{equation*}
\eta(X\bullet Y,Z)=\eta(X,Y\bullet Z), \quad X,Y,Z \in \Gamma(TM),
\end{equation*}
and its Levi-Connection $\nabla$ is together the geometric structure $(M,\bullet,e,\nabla)$ a flat F-manifold.

\item A Dubrovin Frobenius manifold $(M,\bullet,e,\eta, E)$ (with an Euler vector field) is a Dubrovin Frobenius manifold $(M,\bullet,e,\eta)$ (without an Euler vector field) together an Euler vector field $E$ such that
\begin{equation*}
Lie_E\eta = (2-d)\eta, 
\end{equation*}
for some $d\in \mathbb{C}$.
\end{enumerate}
\end{definition}

Given a flat F-manifold \((M, \bullet, e, \nabla)\) of dimension \(n\), we always have locally defined flat coordinates \(t^1, \dots, t^n\), due to the flatness of \(\nabla\). Moreover, let \(c_{\alpha\beta}^{\gamma}(t)\) be the holomorphic structure constants of the corresponding family of commutative and associative algebras in flat coordinates. That is, given the vector fields  
\[
\partial_{\alpha} = \frac{\partial}{\partial t^{\alpha}}, \quad \alpha = 1, \dots, n,
\]  
the structure constants \(c_{\alpha\beta}^{\gamma}(t)\) gives the multiplication table of the Frobenius algebra as follows 
\begin{equation*}
\partial_{\alpha} \bullet \partial_{\beta} = c_{\alpha\beta}^{\gamma} \partial_{\gamma}.
\end{equation*}

The integrability condition (\ref{Hertling integrability condition}) in flat coordinates can then be expressed as  
\begin{equation*}
\partial_{\delta} c_{\alpha\beta}^{\gamma} = \partial_{\alpha} c_{\delta\beta}^{\gamma}.
\end{equation*}

This implies that there exists a local holomorphic function \(\omega_{\beta}^{\gamma}(t)\) such that  
\begin{equation*}
c_{\alpha\beta}^{\gamma} = \partial_{\alpha} \omega_{\beta}^{\gamma}.
\end{equation*}

Furthermore, due to the commutativity condition of the algebra, \(c_{\alpha\beta}^{\gamma}(t)\) is symmetric in \(\alpha\) and \(\beta\), i.e.,  
\begin{equation*}
\partial_{\alpha} \omega_{\beta}^{\gamma} = \partial_{\beta} \omega_{\alpha}^{\gamma}.
\end{equation*}

As a consequence, there exists a ocal vector-valued function \(F^{\gamma}(t)\) such that  
\begin{equation*}
c_{\alpha\beta}^{\gamma} = \partial_{\alpha} \partial_{\beta} F^{\gamma}.
\end{equation*}

In this setting, a flat F-manifold can be viewed as a domain satisfying a system of partial differential equations (PDEs), which was introduced in \cite{Alexander} as the Oriented WDVV equations.

\section{Landau Ginzburg superpotential}

In this section, we review key aspects of Dubrovin Frobenius manifolds, including the Dubrovin connection, semisimple Dubrovin Frobenius manifolds, and the Landau-Ginzburg superpotential, which will be essential for proving the main result later.

In the analytic theory of Dubrovin Frobenius manifold, there exist two flat connection. The 1st structure connection is  called Dubrovin connection and it is defined below

\begin{definition}\cite{B. Dubrovin2},\cite{B. Dubrovin3}
Consider the following deformation of the Levi-Civita connection defined on a Dubrovin Frobenius manifold $M$:
\begin{equation*}
\tilde{\nabla}_{u}v = \nabla_{u}v + zu\bullet v, \quad u,v\in \Gamma(TM),
\end{equation*}
where $\nabla$ represents the Levi-Civita connection of the metric $\eta$, $\bullet$ denotes the Frobenius product, and $z\in \mathbb{CP}^1$. The Dubrovin connection defined in $M\times\mathbb{CP}^1$ is then given by:
\begin{equation}\label{Dubrovin Connection}
\begin{split}
\tilde{\nabla}_{u}v &= \nabla_{u}v + zu\bullet v,\\
\tilde{\nabla}_{\frac{d}{dz}}\frac{d}{dz} &= 0, \quad \tilde{\nabla}_{v} \frac{d}{dz} = 0,\\
\tilde{\nabla}_{\frac{d}{dz}}v &= \partial_{z}v + E\bullet v - \frac{1}{z}\mu(v).
\end{split}
\end{equation}
Here, $\mu$ is a diagonal matrix given by:
\begin{equation*}
\begin{split}
\mu_{\alpha\beta} = \left(q_{\alpha}-\frac{d}{2}\right)\delta_{\alpha\beta}.
\end{split}
\end{equation*}
\end{definition}

The deformation of Levi-Civita connection (\ref{Dubrovin Connection}) is again a flat connection. In Saito flat coordinates, the Dubrovin connection flat coordinate system, i.e, the solution of 
\begin{equation*}
\tilde \nabla d\tilde t=0,
\end{equation*}
 can be written as
 \begin{equation}\label{flat section Dubrovin connection}
 \begin{split}
&\left(\tilde \nabla_{\alpha} \omega\right)_{\beta}=\partial_{\alpha}\omega_{\beta}-zc_{\alpha\beta}^{\gamma}\omega_{\gamma}=0,\\
&\left(\tilde \nabla_{\frac{d}{dz}} \omega\right)_{\beta}=\partial_{z}\omega_{\beta}-E^{\sigma}c_{\sigma\beta}^{\gamma}\omega_{\gamma}-\frac{\mu_{\beta}}{z}\omega_{\beta}=0,\\
\end{split}
\end{equation}
 where $\omega=d\tilde t=\omega_{\alpha}dt^{\alpha}$. 

Consider the multiplication by the Euler vector field:
\begin{equation}\label{multiplication by the Euler vector field endomorphism}
E\bullet: \Gamma(TM) \mapsto \Gamma(TM), \quad X \in \Gamma(TM) \mapsto E\bullet X \in \Gamma(TM)
\end{equation}
Such an endomorphism gives rise to a bilinear form in the sections of the cotangent bundle of $M$ as follows. Consider $x = \eta(X, )$, $y = \eta(Y, ) \in \Gamma(T^*M)$ where $X, Y \in \Gamma(TM)$. An induced Frobenius algebra is defined on $\Gamma(T^*M)$ by:
\begin{equation*}
x \bullet y = \eta(X \bullet Y)
\end{equation*}

\begin{definition}\cite{B. Dubrovin2},\cite{B. Dubrovin3}
The intersection form is a bilinear pairing in $\Gamma(T^*M)$ defined by:
\begin{equation*}
(\omega_1, \omega_2)^* = \iota_E(\omega_1 \bullet \omega_2)
\end{equation*}
where $\omega_1, \omega_2 \in \Gamma(T^*M)$, and $\bullet$ is the induced Frobenius algebra product in $\Gamma(T^*M)$. The intersection form will be denoted by $g^*$. 
\end{definition}

In the flat coordinates of the Saito metric, is given by
\begin{equation}\label{intersection form generic in flat coordinates}
g^{\alpha\beta} = E^{\epsilon}\eta^{\alpha\mu}\eta^{\beta\lambda}c_{\epsilon\mu\lambda}.
\end{equation}

Recall that a point in a Dubrovin Frobenius manifold is called semisimple if the Frobenius algebra in $T_pM$ is semisimple. It's worth noting that semisimplicity constitutes an open condition. The Dubrovin Frobenius structure around semisimple points becomes rather simple. Specifically, the Frobenius algebra becomes trivial. Moreover, both the Saito metric and the endomorphism resulting from the multiplication by the Euler vector field (\ref{multiplication by the Euler vector field endomorphism}) are diagonal around such a point. More precisely,

\begin{proposition}\cite{B. Dubrovin2},\cite{B. Dubrovin3}
Let $(u_1,u_2,..,u_n)$ be pairwise distinct roots of the characteristic equation
\begin{equation}\label{spectral curve}
det(g^{\alpha\beta} -u\eta^{\alpha\beta} ) = 0.
\end{equation}
Then, the roots $(u_1(t),u_2(t),..,u_n(t))$ can serve as local coordinates, which are called canonical coordinates. In these coordinates, the Frobenius multiplication, Saito metric, unit vector field and Euler vector field can be written as 
\begin{equation}\label{Dubrovin Frobenius structure in canonical coordinates 1}
\begin{split}
&\frac{\partial}{\partial u_i}\bullet \frac{\partial}{\partial u_j}=\delta_{ij}\frac{\partial}{\partial u_i},\quad \eta=\sum_{i=1}^{n}\psi_{i1}^2 (du_i)^2, \quad e=\sum_{i=1}^n\frac{\partial}{\partial u_i},\quad E=\sum_{i=1}^n u_i\frac{\partial}{\partial u_i},\\
\end{split}
\end{equation}
where the matrix $\Psi=\left( \psi_{i\alpha}\right)$ is given by
\begin{equation*} 
\psi_{i\alpha}=\psi_{i1}\frac{\partial u_i}{\partial t^{\alpha}}.
\end{equation*}

\end{proposition}

At this stage, we can define a Landau Ginzburg superpotental, which can be found in definition 5.7 of  \cite{B. Dubrovin3}.

\begin{definition}\label{Landau-Ginzburg superpotential definition}\cite{B. Dubrovin3}
Let $D$ be an open domain of a Riemann surface. A Landau-Ginzburg superpotential associated with a Dubrovin Frobenius manifold $M$ of dimension $n$ consists of a function $\lambda(p,u)$ on $D \times M$ and an Abelian differential $\phi$ in $D$ satisfying
\begin{itemize}
\item The critical values of  $\lambda(p,u)$ are the canonical coordinates $(u_1,..,u_n)$. In other words, the canonical coordinates $(u_1,u_2,. ,u_n)$ are defined by the following system
\begin{equation*}
\begin{split}
&\lambda(p_i)=u_i,\\
&\frac{d \lambda}{dp}\left( p_i\right)=0.
\end{split}
\end{equation*}
\item For some cycles $Z_1,..,Z_n$ in $D$ the integrals 
\begin{equation}\label{Mirror symmetry}
\tilde t_j(z,u)=\frac{1}{z^{\frac{1}{2}}}\int_{Z_j} e^{z\lambda(p)}\phi, \quad j=1,..,n.
\end{equation}
converges and give a system of independent flat coordinates for the Dubrovin Connection $\tilde\nabla$ in canonical coordinates, i.e., the matrix
\begin{equation*}
\begin{split}
Y=\Psi\eta^{-1}\omega:=\left(\psi_{i\alpha}\eta^{\alpha\beta}\partial_{\beta}\tilde t   \right)
\end{split}
\end{equation*}
is a solution of the following system
\begin{equation}\label{Dubrovin connection in canonical coordinates 1}
\begin{split}
\frac{\partial Y}{\partial u_i}&=\left( zE_i+V_i \right)Y,\\
\frac{dY}{\partial z}&=\left( U+\frac{V}{z} \right)Y,\\
\end{split}
\end{equation}
where 
\begin{equation}\label{U,V, Vi in canonical coordinates}
\begin{split}
U=\Psi\mathcal{U}\Psi^{-1}, \quad V=\Psi\mu\Psi^{-1}, \quad E_i=\left( \delta_{ij}\delta_{ik}  \right), \quad V_i:=\frac{\partial \Psi}{\partial u_i}\Psi^{-1}.
\end{split}
\end{equation}

\item 
The following expressions for the tensors  Saito metric $\eta$,  intersection form $g^{*}$ and the structure constants $c$ holds true
\begin{equation}\label{residue expression for eta, intersection form and structure constants}
\begin{split}
\eta_{ij}&=\sum \res_{d\lambda=0}   \frac{ \partial_i\lambda  \partial_j\lambda }{d_p\lambda        }\phi, \\
g_{ij}&=\sum \res_{d\lambda=0}  \frac{ \partial_i\log\lambda  \partial_j\log\lambda }{d_p\log\lambda        } \phi,\\
c_{ijk}&=\sum \res_{d\lambda=0}   \frac{ \partial_i\lambda  \partial_j\lambda  \partial_k\lambda  }{d_p\lambda        }\phi.\\
\end{split}
\end{equation}
\end{itemize}

\end{definition}

\section{Hurwitz Dubrovin Frobenius manifolds}

In this section, we introduce the definition of Hurwitz space and the prescription of the Dubrovin Frobenius structure given the data of a Hurwitz space and a suitable choice of Abelian differentials on it. The main reference of this section are \cite{B. Dubrovin2} and \cite{V. Shramchenko}.\\

\begin{definition}
The Hurwitz space $H_{g,n_0 ,...,n_m}$ is the moduli space of curves $C_g$ of genus g
endowed with a N branched covering, $ \lambda: C_g\mapsto \mathbb{C}P ^1$ of $\mathbb{C}P^1$ with $m + 1$ branching points
over $\infty\in \mathbb{C}P^1$ of branching degree $n_i + 1, i = 0, . . . , m.$
\end{definition}

\begin{definition}
Two pairs $(C_g,\lambda)$ and $(\tilde C_g,\tilde\lambda)$ are said Hurwitz-equivalent if there exist an analytic isomorphic $F:C_g\mapsto \tilde C_g$ such that 
\begin{equation*}
\lambda\circ F =\tilde\lambda.
\end{equation*}
\end{definition}

Roughly speaking, Hurwitz spaces $H_{g,n_0 ,...,n_m}$ are moduli spaces of meromorphic functions which realise a Riemann surface of genus $g$ $C_g$ as covering over  $\mathbb{C}P ^1$ with a fixed ramification profile. 

\textbf{Example 1:}\\
The Hurwitz space is $H_{0,n}$ is  given by
\begin{equation*}
H_{0,n}=\left\{ \lambda(p,x_0,x_1,x_2,..,x_n)=\prod_{i=0}^n (p-x_i): \sum_{i=0}^n x_i=0\right\}
\end{equation*}

\textbf{Example 2:}\\
The Hurwitz space is $H_{0,n-1,0}$ is given by
\begin{equation*}
H_{0,n}=\left\{ \lambda(p,a_2,a_3,..,a_{n+1},a_{n+2})=p^n+ a_2p^{n-2}+...+a_np+a_{n+1}+\frac{a_{n+2}}{p}   \right\}
\end{equation*}

\textbf{Example 3:}\\
The Hurwitz space is $H_{1,n}$ is given by
\begin{equation*}
H_{1,n}=\left\{ \lambda(p,v_0,v_1,..,v_n,\tau)=\frac{\prod_{i=0}^n \theta_1(p-v_i,\tau)e^{-2\pi i u}}{\theta_1^{n+1}(v,\tau)}: \sum_{i=0}^n v_i=0   \right\}
\end{equation*}

\textbf{Example 4:}\\
The Hurwitz space is $H_{1,n-1,0}$ is given by
\begin{equation*}
H_{1,n-1,0}=\left\{ \lambda(p,v_0,..,v_n,v_{n+1},\tau)=\frac{\prod_{i=0}^n \theta_1(p-v_i,\tau)e^{-2\pi i u}}{\theta_1^{n}(v,\tau)\theta_1(v+(n+1)v_{n+1},\tau)}:\sum_{i=0}^n v_i=-(n+1)v_{n+1}   \right\}
\end{equation*}

The covering $\tilde H = \tilde H_{g,n_0 ,...,n_m} $ consist of the sets
\begin{equation*}
(C_g; \lambda; k_0, . . . , k_m; a_1, . . . , a_g, b_1, . . . , b_g) \in \tilde H_{g,n_0 ,...,n_m}
\end{equation*}
where $a_1, . . . , a_g, b_1, . . . , b_g \in H_1(C_g,\mathbb{Z})$ are the canonical symplectic basis, and  $k_0, . . . , k_m $ are  roots of $\lambda$ near  $\infty_0 ,\infty_1 , . . . , \infty_m$  of the orders
$n_0 + 1, n_1 + 1, . . . , n_m + 1.$ resp., 
\begin{equation*}
k^{-n_i-1}_i (P) = \lambda(P),\quad  \text{P near $\infty_i$}.
\end{equation*}

Over the space $\tilde H_{g,n_0 ,...,n_m}$, it is possible to introduce a Dubrovin-Frobenius structure by taking as canonical coordinates $(u_1,u_2,..u_n)$ as solution of the following system:
\begin{equation*}
\left\{u_i=\lambda(P_i), \frac{d\lambda}{dp}(P_i)=0\right\}
\end{equation*}
The Dubrovin-Frobenius structure is specified by the following objects:
\begin{equation}\label{Multiplication}
 \textnormal{multiplication}\quad  \partial_i\bullet\partial_j=\delta_{ij}\partial_i,\textnormal{where}\quad \partial_i=\frac{\partial}{\partial u_i},
\end{equation}
\begin{equation}\label{Euler vector field}
 \textnormal{Euler vector field}\quad E=\sum_i u_i\partial_i,
\end{equation}
\begin{equation}\label{Unit vector field}
 \textnormal{unit vector field}\quad e=\sum_i \partial_i,
\end{equation}
and the metric $\eta$ defined by the formula
\begin{equation}\label{metric}
ds^2_{\phi}=\sum \text{res}_{P_i}\frac{\phi^2}{d\lambda}(du_i)^2 ,
\end{equation}
where $\phi$ is some primary differential of the underlying Riemann surface $C_g$.
Note that the Dubrovin-Frobenius manifold structure depends on the meromorphic function $\lambda$, and on the primary differential $\phi$. The list of possible  primary differential $\phi$ are described  in chapter 5 of \cite{B. Dubrovin2}.\\

Consider a multivalued function $p$ on C by taking the integral of $\phi$

\begin{equation*}
p(P)=v.p\int_{\infty_0}^P \phi
\end{equation*}
The principal value is defined by omitting the divergent part, when necessary, because $\phi$ may be divergent at $\infty_0$, as function of the local parameter $k_0$. 

Consider the following differential 
\begin{equation*}
\phi=dp
\end{equation*}
Let $\tilde H_{\phi}$ be the open domain in $\tilde H$ specifying by the condition
\begin{equation*}
\phi(P_i)\neq 0.
\end{equation*}

\begin{theorem}\label{flatcoord}\cite{B. Dubrovin2}
For any primary differential $\phi$ of the list in \cite{B. Dubrovin2} the multiplication (\ref{Multiplication}), the
unity (\ref{Unit vector field}), the Euler vector field (\ref{Euler vector field}), and the metric $\ref{metric}$ determine a structure of Dubrovin Frobenius manifold on $\tilde H_{\phi}$. The corresponding flat coordinates $t_A, A = 1, . . . ,N$ consist of the five parts
\begin{equation}
t_A=(t^{i,\alpha},i=0,..m,\alpha=1,..,n_i;p^i,q^i,i=1,..,m,r^i,s^i,i=1,..g)
\end{equation}
where
\begin{equation}\label{flat coordinates of Saito metric as periods}
t^{i,\alpha}=\text{res}_{\infty_i} k_i^{-\alpha}pd\lambda             \quad  i=0,..m,\alpha=1,..,n_i
\end{equation}
\begin{equation}
p^{i}=v.p\int_{\infty_0}^{\infty_i}dp             \quad  i=1,..m.
\end{equation}
\begin{equation}
q^{i}=-\text{res}_{\infty_i} \lambda dp         \quad  i=1,..m.
\end{equation}
\begin{equation}
r^{i}=\int_{b_i}dp             \quad  i=1,..g.
\end{equation}
\begin{equation}
s^{i}=-\frac{1}{2\pi i}\int_{a_i} \lambda dp         \quad  i=1,..g.
\end{equation}
\end{theorem}

\section{Alcolado Open WDVV equations formula}

The main objective of this section is to revisit the derivation of the Alcolado Open WDVV formula (\ref{Adam identity 1 equation1}). To achieve this, we first recall the proof of Theorem 1.2 from \cite{Alcolado Adam}.\\

Let  $F(t^1, t^2, .., t^n)$ be a (closed) WDVV solution and the correspondent pair
\begin{equation*}
\Omega=\Omega(t^1,..,t^n), \quad \frac{\partial\Omega}{\partial p}\neq 0,
\end{equation*}

such that $(F, \Omega)$ satisfy Open WDVV equations. More explicitly, 

\begin{equation*}
c_{\alpha\beta}^{\gamma}=\eta^{\gamma\mu}\frac{\partial^3 F}{\partial t^{\alpha}\partial t^{\beta}\partial t^{\gamma}}, \quad \Omega_{\alpha\beta}=\frac{\partial^2 \Omega}{\partial t^{\alpha}\partial t^{\beta}}, \quad \Omega_{p\beta}=\frac{\partial^2 \Omega}{\partial p\partial t^{\beta}}, \quad \Omega_{pp}=\frac{\partial^2 \Omega}{\partial p\partial p},
\end{equation*}
then the Open WDVV associative equation become
\begin{equation}\label{Adam identity 0}
\begin{split}
c_{\alpha\beta}^{\delta}\Omega_{\gamma\delta}+\Omega_{\alpha\beta}\Omega_{\gamma p}&=c_{\gamma\beta}^{\delta}\Omega_{\alpha\delta}+\Omega_{\gamma\beta}\Omega_{\alpha p},\\
c_{\alpha\beta}^{\delta}\Omega_{p \delta}+\Omega_{\alpha\beta}\Omega_{p p}&=\Omega_{\alpha p}\Omega_{\beta p},\\
\end{split}
\end{equation}
Assuming that there exist $\lambda(p,t^1,t^2,..,t^n)$ such that $\lambda = \Omega_p$, then the second equation of (\ref{Adam identity 0}) can be written as
\begin{equation}\label{Adam identity 01}
\begin{split}
\Omega_{\alpha\beta}&=\frac{\lambda_{\alpha }\lambda_{\beta}-c_{\alpha\beta}^{\delta}\lambda_{\delta}}{\lambda_{p}} .\\
\end{split}
\end{equation}
Moreover, the first equation of (\ref{Adam identity 0}) is automatic satisfied due to (\ref{Adam identity 01}). Indeed,
\begin{equation*}
\begin{split}
c_{\alpha\beta}^{\delta}\left(\frac{\lambda_{\gamma }\lambda_{\delta}-c_{\gamma\delta}^{\epsilon}\lambda_{\epsilon}}{\lambda_{p}}  \right)+\left(\frac{\lambda_{\alpha }\lambda_{\beta}-c_{\alpha\beta}^{\delta}\lambda_{\delta}}{\lambda_{p}}\right)\lambda_{\gamma}&=-\frac{c_{\alpha\beta}^{\delta}c_{\gamma\delta}^{\epsilon}\lambda_{\epsilon}}{\lambda_p}+\frac{\lambda_{\alpha}\lambda_{\beta}\lambda_{\gamma}}{\lambda_p},\\
&=-\frac{c_{\gamma\beta}^{\delta}c_{\alpha\delta}^{\epsilon}\lambda_{\epsilon}}{\lambda_p}+\frac{\theta_{\alpha}\lambda_{\beta}\lambda_{\gamma}}{\lambda_p},\\
&=c_{\gamma\beta}^{\delta}\left(\frac{\lambda_{\alpha }\lambda_{\delta}-c_{\alpha\delta}^{\epsilon}\lambda_{\epsilon}}{\lambda_{p}}  \right)+\left(\frac{\lambda_{\gamma }\lambda_{\beta}-c_{\gamma\beta}^{\delta}\lambda_{\delta}}{\lambda_{p}}\right)\lambda_{\alpha}.
\end{split}
\end{equation*}

Summarising,

\begin{theorem}\cite{Alcolado Adam}
For a fixed solution $F(t^1, t^2, ..., t^n)$ of the closed WDVV equations (\ref{closed WDVV equation}). Let $\lambda= \lambda(p,t^1, . . . , t^n)$ be a function satisfying $\frac{\partial \lambda}{\partial p}\neq  0$ and
 \begin{equation}\label{Adam identity 1 equation}
\partial_{\alpha}\partial_{\beta}\left(\int \lambda(p)dp\right)= \frac{ \partial_{\alpha}\lambda \partial_{\beta}\lambda-c_{\alpha\beta}^{\gamma}\partial_{\gamma}\lambda}{ \partial_{p}\lambda},
\end{equation}
where $c_{\alpha\beta}^{\gamma}=\frac{\partial^3 F}{\partial t^{\alpha}\partial t^{\beta}\partial t^{\gamma}}$.
Then the pair $(F,\int\lambda dp)$ do satisfy an Open WDVV equations (\ref{Open WDVV equation 0}) with respect the closed WDVV solution $F(t^1, t^2, ..., t^n)$.

\end{theorem}

\section{From LG superpotential  to Open WDVV solution}

This section is devoted to prove the theorem \ref{Main theorem Almeida}.

\begin{lemma}\label{Auxiliar lemma Generalized Adam Id 1}
 Let M be a n-dimensional Dubrovin Frobenius manifold with flat coordinates $t^{\alpha}$  and $\lambda(p)$ its Landau-Ginzburg superpotential. Then, the function 
 \begin{equation}\label{Bounded rhs of Adam id}
\partial_{\alpha}\partial_{\beta}\lambda-\frac{\partial}{\partial p}\left[ \frac{ \partial_{\alpha}\lambda \partial_{\beta}\lambda-c_{\alpha\beta}^{\gamma}\partial_{\gamma}\lambda}{ \partial_{p}\lambda}  \right]
\end{equation}
is holomorphic at all critical points of $\lambda$.

\end{lemma}

\begin{proof}
Recall that any LG superpotential $\lambda$ has the following holomorphic behaviour near its critical points
\begin{equation}
\lambda(p)=u_i+\eta_{ii}(p-u_i)^2+O(p-u_i)^3).
\end{equation}
Therefore, the function $\partial_{\alpha}\partial_{\beta}\lambda(p)$ is holomorphic around $p=u_i$. The other piece of the function (\ref{Bounded rhs of Adam id}) has the following behaviour near $p=u_i$
 \begin{equation}\label{Bounded rhs of Adam id 10}
 \begin{split}
\frac{ \partial_{\alpha}\lambda \partial_{\beta}\lambda-c_{\alpha\beta}^{\gamma}\partial_{\gamma}\lambda}{ \partial_{p}\lambda}&=  \frac{ \partial_{\alpha}u_i \partial_{\beta}u_i-c_{\alpha\beta}^{\gamma}\partial_{\gamma}u_i}{ \eta^{-1}_{ii}(p-u_i)}+ O(1),\\
&=  \frac{ \partial_{\alpha}u_i \partial_{\beta}u_i-c_{jk}^{l}\partial_{\alpha}u_j\partial_{\beta}u_k\partial_{i}t^{\gamma}\partial_{\gamma}u_i}{ \eta^{-1}_{ii}(p-u_i)}+ O(1),\\
&=  \frac{ \partial_{\alpha}u_i \partial_{\beta}u_i-\partial_{\alpha}u_j\partial_{\beta}u_k\delta_i^l\delta_{j}^{l}\delta_{k}^l}{ \eta^{-1}_{ii}(p-u_i)}+ O(1),\\
&= O(1).\\
\end{split}
\end{equation}
Since in canonical coordinates the Frobenius multiplication has the following form $c_{ij}^k=\delta_i^k\delta_j^k$. Lemma proved.

\end{proof}

\begin{lemma}
Let M be a n-dimensional Dubrovin Frobenius manifold with flat coordinates $t^{\alpha}$  and $\lambda(p)$ its Landau-Ginzburg superpotential. Then,
 
 \begin{equation}\label{second derivative of Hodge variation integral version}
\int_{Z_j}\left[\partial_{\alpha}\partial_{\beta}\lambda-\frac{\partial}{\partial p}\left[ \frac{ \partial_{\alpha}\lambda \partial_{\beta}\lambda-c_{\alpha\beta}^{\gamma}\partial_{\gamma}\lambda}{ \partial_{p}\lambda}  \right]\right] dp=0,
\end{equation}

where $Z_j$ forms a basis for the homology of $\Lambda^{*}(z)$.
\end{lemma}

\begin{proof}

The Dubrovin Connection flat section (\ref{flat section Dubrovin connection}) are given by

\begin{equation}\label{Dubrovin connection in flat coordinates}
\begin{split}
\partial_{\alpha}\partial_{\beta}\tilde t-zc_{\alpha\beta}^{\gamma}\partial_{\gamma}\tilde t=&0,\\
\partial_{z}\partial_{\beta}\tilde t+\frac{\mu}{z}\partial_{\beta}\tilde t=0.
\end{split}
\end{equation}
Substituting (\ref{Mirror symmetry}) in the first equation of  (\ref{Dubrovin connection in flat coordinates}), we obtain
\begin{equation}\label{Adam identity part 1}
\begin{split}
\partial_{\alpha}\partial_{\beta}\tilde t-zc_{\alpha\beta}^{\gamma}\partial_{\gamma}\tilde t=&\frac{1}{\sqrt{z}}\int_{Z_j} \left(z\partial_{\alpha}\partial_{\beta}\lambda +z^2\partial_{\alpha}\lambda\partial_{\beta}\lambda-z^2c_{\alpha\beta}^{\gamma}\partial_{\gamma}\lambda\right)e^{z\lambda(p)}dp,\\
&=z^{\frac{1}{2}}\int_{Z_j} \partial_{\alpha}\partial_{\beta}\lambda e^{z\lambda(p)}dp+z^{\frac{3}{2}}\int_{Z_j} \left(\partial_{\alpha}\lambda\partial_{\beta}\lambda-c_{\alpha\beta}^{\gamma}\partial_{\gamma}\lambda\right)e^{z\lambda(p)}dp,\\
&=z^{\frac{1}{2}}\int_{Z_j} \partial_{\alpha}\partial_{\beta}\lambda e^{z\lambda(p)}dp+z^{\frac{3}{2}}\int_{C_j} \left(\frac{\partial_{\alpha}\lambda\partial_{\beta}\lambda-c_{\alpha\beta}^{\gamma}\partial_{\gamma}\lambda}{d_p\lambda}\right)e^{z\lambda}d\lambda,\\
&=z^{\frac{1}{2}}\int_{Z_j} \partial_{\alpha}\partial_{\beta}\lambda e^{z\lambda(p)}dp-z^{\frac{1}{2}}\int_{Z_j} \frac{d}{dp}\left(\frac{\partial_{\alpha}\lambda\partial_{\beta}\lambda-c_{\alpha\beta}^{\gamma}\partial_{\gamma}\lambda}{d_p\lambda}\right)e^{z\lambda(p)}dp,\\
&=z^{\frac{1}{2}}\int_{Z_j}\left[ \partial_{\alpha}\partial_{\beta}\lambda - \frac{d}{dp}\left(\frac{\partial_{\alpha}\lambda\partial_{\beta}\lambda-c_{\alpha\beta}^{\gamma}\partial_{\gamma}\lambda}{d_p\lambda}\right)\right]e^{z\lambda(p)}dp,\\
&=0.
\end{split}
\end{equation}

Lemma proved.

\end{proof}

\begin{lemma}\label{Auxiliar lemma 2 Laurent Taylor}
 Let M be a n-dimensional Dubrovin Frobenius manifold with flat coordinates $t^{\alpha}$  and $\lambda(p)$ its Landau-Ginzburg superpotential. Then, the function
  \begin{equation}\label{Laurent expansion of Adam Id}
\partial_{\alpha}\partial_{\beta}\lambda-\frac{\partial}{\partial p}\left[ \frac{ \partial_{\alpha}\lambda \partial_{\beta}\lambda-c_{\alpha\beta}^{\gamma}\partial_{\gamma}\lambda}{ \partial_{p}\lambda}  \right]
\end{equation}
 is holomorphic at the pre image of $\infty$. ie.. $\lambda^{-1}(\infty)=\infty_0,\infty_1,..,\infty_m$.

\end{lemma}

\begin{proof}
From proposition 3 of \cite{V. Shramchenko 1}, we have that a- and b-cycles on the underlying  Riemann surface $\Sigma_g$ and contours $\gamma_k$ encircling the poles $\infty_0,..\infty_m$ are elements of a suitable basis for the homology of $\Lambda^{*}(z)$. Therefore,

\begin{equation}\label{Adam identity part 12}
\begin{split}
\int_{\alpha_j}\left[ \partial_{\alpha}\partial_{\beta}\lambda - \frac{d}{dp}\left(\frac{\partial_{\alpha}\lambda\partial_{\beta}\lambda-c_{\alpha\beta}^{\gamma}\partial_{\gamma}\lambda}{d_p\lambda}\right)\right]dp=0,\quad j=1,,.g,\\
\int_{\beta_j}\left[ \partial_{\alpha}\partial_{\beta}\lambda - \frac{d}{dp}\left(\frac{\partial_{\alpha}\lambda\partial_{\beta}\lambda-c_{\alpha\beta}^{\gamma}\partial_{\gamma}\lambda}{d_p\lambda}\right)\right]dp=0,\quad j=1,,.g,\\
\int_{\gamma_k}\left[ \partial_{\alpha}\partial_{\beta}\lambda - \frac{d}{dp}\left(\frac{\partial_{\alpha}\lambda\partial_{\beta}\lambda-c_{\alpha\beta}^{\gamma}\partial_{\gamma}\lambda}{d_p\lambda}\right)\right]dp=0.\quad k=0,1,..,m.
\end{split}
\end{equation}
As a consequence, the function  (\ref{Laurent expansion of Adam Id}) is invariant $\alpha$ and $\beta$ cycles and it does not have simple poles around $\infty_0,..\infty_m$. In order to proof that   (\ref{Laurent expansion of Adam Id})  is holomorphic around $\infty_0,..\infty_m$ it is enough to check its higher order poles.\\

Consider the local change of  coordinate around some generic pole
\begin{equation}\label{multivalued variable in the Hurwitz space}
p(k)=\frac{1}{n+1}\sum_{\alpha=0}^n t^{n+1-\alpha}k^{\alpha+1} +O(k^{n+2}), \quad \text{where} \quad \lambda=\frac{1}{k^{n+1}}.
\end{equation}
Then,
\begin{equation}
\frac{\partial p}{\partial t^{\beta}}=\frac{k^{n+1-\beta}}{n+1}  +O(k^{n+2})
\end{equation}

Recall the Thermodynamic identity 
\begin{equation}
\frac{\partial \lambda}{\partial t^{\alpha}}=-\frac{\partial p}{\partial t^{\alpha}}\frac{\partial \lambda}{\partial p}.
\end{equation}
Then,

\begin{equation}\label{part structure constant}
\frac{c_{\alpha\beta}^{\gamma}\frac{\partial \lambda}{\partial t^{\gamma}}}{\frac{\partial \lambda}{\partial p}}=-c_{\alpha\beta}^{\gamma}\frac{\partial p}{\partial t^{\gamma}}
\end{equation}
Substituting (\ref{part structure constant}) in 
\begin{equation}\label{part structure constant 2}
\frac{\partial}{\partial p}\left(\frac{c_{\alpha\beta}^{\gamma}\frac{\partial \lambda}{\partial t^{\gamma}}}{\frac{\partial \lambda}{\partial p}} \right)dp,
\end{equation}
we obtain,
\begin{equation}\label{part structure constant 3}
\begin{split}
\frac{\partial}{\partial p}\left(\frac{c_{\alpha\beta}^{\gamma}\frac{\partial \lambda}{\partial t^{\gamma}}}{\frac{\partial \lambda}{\partial p}} \right)dp&=\frac{\partial}{\partial p}\left( -c_{\alpha\beta}^{\gamma}\frac{\partial p}{\partial t^{\gamma}}
 \right)dp,\\
 &=\frac{\partial k}{\partial p}\frac{\partial}{\partial k}\left( -c_{\alpha\beta}^{\gamma}\frac{\partial p}{\partial t^{\gamma}}
 \right)\frac{\partial p}{\partial k}dk,\\
  &=\frac{\partial}{\partial k}\left( -c_{\alpha\beta}^{\gamma}\frac{\partial p}{\partial t^{\gamma}}
 \right)dk,\\
   &=\frac{\partial}{\partial k}\left( -\frac{c_{\alpha\beta}^{\gamma}}{n+1} k^{n+1-\gamma} +O(k^{n+2})
 \right)dk,\\
  &=\left(-\frac{\left(n+1-\gamma\right)}{n+1}c_{\alpha\beta}^{\gamma} k^{n-\gamma} +O(k^{n+2}) \right)dk.\\
 \end{split}
\end{equation}
Therefore, the function (\ref{part structure constant 2}) is holomorphic in $k$.

Recall the formula

\begin{equation}\label{flat coordinates of Saito metric as periods n root}
t^{\gamma}=\res_{p=\infty} \lambda(p)^{\frac{n+1-\gamma}{n+1}} dp,
\end{equation}
which is equivalent to the equation (\ref{flat coordinates of Saito metric as periods}).
Differentiating (\ref{flat coordinates of Saito metric as periods n root}) we obtain,

\begin{equation}
\begin{split}
0=\frac{\partial ^2t^{\gamma}}{\partial t^{\alpha}\partial t^{\beta}}&=\res_{p=\infty} \frac{n+1-\gamma}{n+1}\left[ \frac{\partial ^2\lambda}{\partial t^{\alpha}\partial t^{\beta}} \lambda(p)^{\frac{-\gamma}{n+1}}-\frac{\gamma}{n+1}\frac{\partial \lambda}{\partial t^{\alpha}}\frac{\partial \lambda}{\partial t^{\beta}} \lambda(p)^{\frac{-n-1-\gamma}{n+1}}\right] dp,\\
&=\res_{p=\infty} \frac{n+1-\gamma}{n+1}\left[ \frac{\partial ^2\lambda}{\partial t^{\alpha}\partial t^{\beta}} \lambda(p)^{\frac{-\gamma}{n+1}}dp \right]-\res_{p=\infty} \frac{n+1-\gamma}{n+1}\left[ \frac{\gamma}{n+1}\frac{\frac{\partial \lambda}{\partial t^{\alpha}}\frac{\partial \lambda}{\partial t^{\beta}} }{\frac{\partial \lambda}{\partial p}}\lambda^{\frac{-n-1-\gamma}{n+1}}\right] d\lambda,\\
&=\res_{p=\infty} \frac{n+1-\gamma}{n+1}\left[ \frac{\partial ^2\lambda}{\partial t^{\alpha}\partial t^{\beta}} \lambda(p)^{\frac{-\gamma}{n+1}}dp \right]-\res_{p=\infty} \frac{n+1-\gamma}{n+1}\left[ \frac{\partial}{\partial p}\left[\frac{\frac{\partial \lambda}{\partial t^{\alpha}}\frac{\partial \lambda}{\partial t^{\beta}} }{\frac{\partial \lambda}{\partial p}}\right]\lambda^{\frac{-\gamma}{n+1}}(p)\right] dp,\\
&=\res_{p=\infty} \frac{n+1-\gamma}{n+1}\left[ \frac{\partial ^2\lambda}{\partial t^{\alpha}\partial t^{\beta}} - \frac{\partial}{\partial p}\left[\frac{\frac{\partial \lambda}{\partial t^{\alpha}}\frac{\partial \lambda}{\partial t^{\beta}} }{\frac{\partial \lambda}{\partial p}}\right]\right]\lambda(p)^{\frac{-\gamma}{n+1}} dp,\\
\end{split}
\end{equation}
which implies that the function 
\begin{equation*}
\frac{\partial ^2\lambda}{\partial t^{\alpha}\partial t^{\beta}} - \frac{\partial}{\partial p}\left[\frac{\frac{\partial \lambda}{\partial t^{\alpha}}\frac{\partial \lambda}{\partial t^{\beta}} }{\frac{\partial \lambda}{\partial p}}\right]
\end{equation*}
is holomorphic in $k$ and since the (\ref{part structure constant 2}) is holomorphic in $k$, we have that (\ref{Laurent expansion of Adam Id}) is holomorphic in k. Lemma proved.

\end{proof}

At this stage, we can prove the main result of this paper which is given as follows:

\begin{theorem}\label{main theorem}
 Let M be a n-dimensional Dubrovin Frobenius manifold with flat coordinates $t^{\alpha}$  and $\lambda(p)$ its Landau-Ginzburg superpotential. Then,
 \begin{equation}\label{second derivative of Hodge variation main theorem}
\partial_{\alpha}\partial_{\beta}\lambda=\frac{\partial}{\partial p}\left[ \frac{ \partial_{\alpha}\lambda \partial_{\beta}\lambda-c_{\alpha\beta}^{\gamma}\partial_{\gamma}\lambda}{ \partial_{p}\lambda}  \right].
\end{equation}
\end{theorem}

\begin{proof}

From Lemma \ref{Auxiliar lemma Generalized Adam Id 1} and Lemma \ref{Auxiliar lemma 2 Laurent Taylor}, we have that the function (\ref{Bounded rhs of Adam id}) is a holomorphic bounded function on a compact Riemann surfaces. Then, it is constant with respect the variable $p$ by Liouville's theorem. Using equation (\ref{Adam identity part 1}), we have 
\begin{equation}
\partial_{\alpha}\partial_{\beta}\lambda - \frac{d}{dp}\left(\frac{\partial_{\alpha}\lambda\partial_{\beta}\lambda-c_{\alpha\beta}^{\gamma}\partial_{\gamma}\lambda}{d_p\lambda}\right)=0.
\end{equation}
Because $\int_{Z_j}e^{z\lambda(p)}dp\neq 0$. Theorem proved.

\end{proof}

\begin{corollary}\label{main corollary}
 Let M be a n-dimensional Dubrovin Frobenius manifold with flat coordinates $t^{\alpha}$  and $\lambda(p)$ its Landau-Ginzburg superpotential. Then,
 \begin{equation}\label{second derivative of Hodge variation corrolary}
\partial_{\alpha}\partial_{\beta}\left(\int\lambda(p)dp\right)= \frac{ \partial_{\alpha}\lambda \partial_{\beta}\lambda-c_{\alpha\beta}^{\gamma}\partial_{\gamma}\lambda}{ \partial_{p}\lambda} .
\end{equation}
\end{corollary}

\begin{proof}
Consider the integration of constant $\tilde\Omega_{\alpha\beta}(t)$ of $\int \lambda(p) dp$ such that
\begin{equation*}
\int \lambda(p) dp+ \tilde\Omega_{\alpha\beta}(t)= \frac{ \partial_{\alpha}\lambda \partial_{\beta}\lambda-c_{\alpha\beta}^{\gamma}\partial_{\gamma}\lambda}{ \partial_{p}\lambda}.
\end{equation*}

Corollary proved.

\end{proof}

\begin{remark}\label{remark triviality of t1}
It is straightforward to see that the identity (\ref{second derivative of Hodge variation corrolary}) holds true for $\alpha=1$ since 
\begin{equation*}
\frac{\partial \lambda}{\partial t^1}=0.
\end{equation*}
Indeed, on one hand
 \begin{equation*}
 \begin{split}
\partial_{1}\partial_{\beta}\left(\int\lambda(p)dp\right)&=\left(\int \partial_{1}\partial_{\beta}\lambda(p)dp\right)=0.
\end{split}
\end{equation*}
On another hand,
 \begin{equation*}
 \begin{split}
\frac{ \partial_{1}\lambda \partial_{\beta}\lambda-c_{1\beta}^{\gamma}\partial_{\gamma}\lambda}{ \partial_{p}\lambda} =\frac{  \partial_{\beta}\lambda-\delta_{\beta}^{\gamma}\partial_{\gamma}\lambda}{ \partial_{p}\lambda} =0.
\end{split}
\end{equation*}

In particular, we have that $\Omega(p,t)=\int \lambda(p)dp$ has property
\begin{equation}\label{Buryak  additional condition}
\frac{\partial^2\Omega}{\partial t^1\partial p}=1, \quad \frac{\partial^2\Omega}{\partial t^1\partial t^{\beta}}=0.
\end{equation}
Note that the condition (\ref{Buryak  additional condition}) is the same condition imposed in the equation (1.8) of \cite{Buryak}.

\end{remark}

\begin{remark}

From corollary \ref{main corollary} and results of the section \ref{Flat F manifold and Dubrovin Frobenius manifold}, we extend any Hurwitz Dubrovin Frobenius manifold $H_{g,n_0,..,n_m}\ni (t^1,t^2,..,t^n)$ to a Flat F manifold $H_{g,n_0,..,n_m}\times C_g \ni (t^1,t^2,..,t^n,p)$ which we denoted by Open Hurwitz Flat F manifold $H_{g,n_0,..,n_m}^o$.

\end{remark}

\section{Examples}

In this section, we provide some low dimensional examples of Open Hurwitz Flat F manifold $H_{g,n_0,..,n_m}^o$. The algorithm to compute the Open WDVV equations are the following: Starting from  the Hurwitz space 
\begin{equation*}
H_{g,n_0,n_1,..,n_m}:=\{ \lambda: C_g\mapsto \mathbb{CP}^1\}\quad  and \quad \phi=dp,
\end{equation*}
 we compute the associated WDVV solution $F(t)$ by using the last  formula of (\ref{residue expression for eta, intersection form and structure constants}) and integrating it. From $F(t)$ and $\lambda(p)$, we compute the open component $\Omega(p,t)$ of the Open WDVV solution  $(F(t),\Omega(p,t))$ by the formula (\ref{second derivative of Hodge variation corrolary}).\\

\subsection{Open  $H_{0,1}$ and $H_{0,2}$}

From the data of the Hurwitz spaces $H_{0,1}$  and $H_{0,2}$ with the holomorphic differential $\phi=dp$ 

\begin{equation*}
\begin{split}
\lambda(p,t^1)=p^2+t^1,\quad \phi=dp,\quad \lambda(p,t^1,t^2)=p^3+t^2p+t^1,\quad \phi=dp.\\
\end{split}
\end{equation*}

We can derive the  Open WDVV solutions associated with the Hurwitz space  $H_{0,1}$  and $H_{0,2}$ with the holomorphic differential $\phi=dp$
\begin{equation}\label{open WDVVlg h01}
\begin{split}
F(t^1)&=\frac{(t^1)^3}{6},\\
\Omega(p,t^1)&=\frac{p^3}{3}+t^1p.
\end{split}
\end{equation}

\begin{equation}\label{Open WDVVlg h02}
\begin{split}
F(t^1,t^2)&=\frac{(t^1)^2t^2}{2}-\frac{(t^2)^4}{72},\\
\Omega(p,t^1,t^2)&=\frac{p^4}{4}+t^2\frac{p^2}{2}+t^1p+\frac{(t^2)^2}{6}.
\end{split}
\end{equation}

The integration constants of (\ref{open WDVVlg h01}) are trivial due to the remark \ref{remark triviality of t1} and the term $\frac{(t^2)^2}{6}$ in the Open WDVV solution (\ref{Open WDVVlg h02}) appears from the difference 
\begin{equation*}
\begin{split}
\frac{ \partial}{\partial^2 t^2}\left( \frac{p^4}{4}+t^2\frac{p^2}{2}+t^1p \right)-\frac{ \partial_{2}\lambda \partial_{2}\lambda-c_{22}^{\gamma}\partial_{\gamma}\lambda}{ \partial_{p}\lambda}.
\end{split}
\end{equation*}
Note that formulas (\ref{open WDVVlg h01}) and (\ref{Open WDVVlg h02}) are already known from Alcolado result's \cite{Alcolado Adam}. In the next sections, we will use the corollary \ref{main corollary} to produce new Open WDDV solutions.

\subsection{Open Quantum Cohomology}

Consider  the Hurwitz space $H_{0,0,0}$  and $H_{0,1,0}$ 

\begin{equation}\label{pre trigonometric functions}
\begin{split}
\lambda(p,a,b)=p+\frac{a}{p-b}, \quad \lambda(p,a,b,c)=\frac{p^2}{2}+c+\frac{a}{p-b},
\end{split}
\end{equation}
with meromorphic differential $\frac{dp}{p-b}$. Then, using the local change of coordinates in $p$ induced by the meromorphic differential $\frac{dp}{p-b}$. i.e
\begin{equation*}
\begin{split}
d\tilde p=\frac{dp}{p-b} \implies \tilde p=\log (p-b).
\end{split}
\end{equation*}
We rewrite (\ref{pre trigonometric functions}), after some suitable change of coordinates in $(a,b,c)$ as 
\begin{equation}\label{trigonometric functions 1}
\begin{split}
\lambda(p,t^1,t^2)&=t^1-2e^{\frac{t^2}{2}}\cos p,\quad \phi=dp,
 \end{split}
\end{equation}
\begin{equation}\label{trigonometric functions 2}
\begin{split}
 \lambda(p,t^1,t^2,t^3)&=e^{2p}+\sqrt{2}t^2e^p+t^1+\frac{e^{t^3-p}}{\sqrt{2}},\quad \phi=dp.
 \end{split}
\end{equation}
The open WDVV solutions associated with the two family of trigonometric functions (\ref{trigonometric functions 1}) and (\ref{trigonometric functions 2}) are given by

\begin{equation}\label{Open Quantum cohomology 1}
\begin{split}
F(t^1,t^2)&=\frac{(t^1)^2t^2}{2}+e^{t^2},\\
\Omega(p,t^1,t^2)&=t^1p-2e^{\frac{t^2}{2}}\sin p
\end{split}
\end{equation}

\begin{equation}\label{Open Quantum cohomology 2}
\begin{split}
F(t^1,t^2,t^3)&=\frac{(t^1)^2t^2}{2}+\frac{(t^2)^2t^1}{2}-\frac{1}{24}(t^2)^4+t^2e^{t^3},\\
\Omega(p,t^1,t^2,t^3)&=\frac{e^{2p}}{2}+\sqrt{2}t^2e^p+t^1p-\frac{e^{t^3-p}}{\sqrt{2}}+\frac{(t^2)^2}{2},\\
\end{split}
\end{equation}
respectively. The equation (\ref{Open Quantum cohomology 1}) have been already obtained in the section  6.3 of \cite{Buryak} by solving the Open WDVV equations and the Quasi-Homogeneous equations explicitly. Moreover, from the best of my knowledge the solution (\ref{Open Quantum cohomology 2}) is new and since Dubrovin Frobenius manifolds associated with Extended affine $A_n$ has a Landau Ginzburg superpotential associated with a Hurwitz space, one could use the corollary \ref{main corollary} to compute the respective Open WDVV solution by just integrating the Landau Ginzburg superpotential  and checking the integration constants. The integration constants of (\ref{Open Quantum cohomology 1}) and (\ref{Open Quantum cohomology 2})  are computed in the Appendix \ref{Appendix}. 
\\

\subsection{Open genus 1 Hurwitz spaces}
From the data of the Hurwitz spaces $H_{1,1}$   

\begin{equation*}
\begin{split}
\lambda(p,t^1,t^2,t^3)&=t^1+(t^2)^2\left[\frac{\partial ^2\log\theta_1(p,t^3)}{\partial ^2p}\right],\\
\end{split}
\end{equation*}
with the holomorphic differential $\phi=dp$. Here $\theta_1(p,\tau)$ is the Jacobi theta 1 function. i.e.

\begin{equation*}
\theta_1(p,\tau)=2\sum_{n=0}^{\infty} (-1)^{n}{\rm e}^{\pi {\rm i}\tau(n+\frac{1}{2})^2 }\sin((2n+1)p).
\end{equation*}

We can derive the  Open WDVV solutions associated with the Hurwitz space  $H_{1,1}$ with the holomorphic differential $\phi=dp$, which is given by

 \begin{equation}\label{Open H11}
 \begin{split}
F\left(t^1,t^2,t^3\right)&=\frac{\left(t^1\right)^2t^3}{2}+\frac{t^1(t^2)^2}{2}-\frac{i\pi}{48}E_2(t^3),\\
\Omega(p,t^1,t^2,t^3)&=t^1p+(t^2)^2\left[\frac{\partial \log\theta_1(p,t^3)}{\partial p}\right].\\
\end{split}
\end{equation}
Here $E_2(\tau)$ is the Eisenstein 2 series. i.e.

\begin{equation*}
E_2(\tau)=1+\frac{3}{\pi^2}\sum_{m\neq 0}\sum_{n=-\infty}^{\infty}\frac{1}{(m+n\tau)^2}.
\end{equation*}

%Furthermore, from the data of the Hurwitz spaces $H_{1,0,0}$   
%
%\begin{equation*}
%\begin{split}
%\lambda(p,t^1,t^2,t^3,t^4)&=t^1+t^2\left[\frac{\theta_1^{\prime}(p-t^3|t^4)}{\theta_1(p-t^3|t^4)}-\frac{\theta_1^{\prime}(p+t^3|t^4)}{\theta_1(p+t^3|t^4)}\right],\\
%\end{split}
%\end{equation*}
%with the holomorphic differential $\phi=dp$.
%
%We can derive the  Open WDVV solutions associated with the Hurwitz space  $H_{1,0,0}$ with the holomorphic differential $\phi=dp$, which is given by
% \begin{equation}\label{Open H100}
% \begin{split}
%F\left(t^1,t^2,t^3,t^4\right)&=\frac{i}{4\pi}\left(t^1\right)^2t^4-2t^1t^2t^3-\left(t^2\right)^2\log\left(t^2\frac{\theta_1^{\prime}\left(0,t^4\right)}{\theta_1\left(2t^3,t^4\right)}\right),\\
%\Omega(p,t^1,t^2,t^3,t^4)&=t^1p+t^2\left[\log\left(\frac{\theta_1(p-t^3|t^4)}{\theta_1(p+t^3|t^4)}\right)\right].\\
%\end{split}
%\end{equation}

The integration constants of (\ref{Open H11})  are computed in the Appendix \ref{Appendix}. 

\section{Conclusion}

The results of this paper provide an efficient method for computing Open WDVV solutions associated with any Hurwitz space. Moreover, these findings suggest potential connections with the Topological Recursion framework \cite{Eynard}, as the data of the Open WDVV solution depends only on the Landau–Ginzburg superpotential and the choice of Abelian differential $\phi$. This connection may also be relevant for constructing open descendant potentials \cite{Alexander}, \cite{Buryak}.

Notably, all the examples provided in this manuscript are related to reflection groups and their extensions \cite{Almeida 1}, \cite{Almeida 2}, \cite{Bertola}, \cite{Bertola 2}, \cite{B. Dubrovin1}, \cite{B. Dubrovin34}, \cite{Zuo}. Therefore, further exploration of the relationship between the monodromy of the Open Hurwitz Flat F-manifold and reflection groups, along with their extensions, could be a promising direction for future research.

\section{Appendix}\label{Appendix}

\subsection{$H_{0,0,0}$ and $H_{0,1,0}$ case }

To compute the integration constants of (\ref{Open Quantum cohomology 1}) is enough to find  $\tilde\Omega(t^2) $ such that 
 \begin{equation*}
\partial^2_{2}\left(t^1p-2e^{\frac{t^2}{2}}\sin p+\tilde\Omega(t^2)   \right)= \frac{ (\partial_2\lambda)^2 -c_{22}^{\gamma}\partial_{\gamma}\lambda}{ \partial_{p}\lambda},
\end{equation*}

which is equal to 

%Then,
%\begin{equation}
%\begin{split}
%\frac{\partial \lambda}{\partial t^1}&=1,\\
%\frac{\partial \lambda}{\partial t^2}&=-e^{\frac{t^2}{2}}\cos p,\\
%\frac{\partial \lambda}{\partial p}&=2e^{\frac{t^2}{2}}\sin p,\\
%\frac{\partial ^2\lambda}{\partial ^2t^2}&=-\frac{1}{2}e^{\frac{t^2}{2}}\cos p,\\
%\end{split}
%\end{equation}

 \begin{equation}
 \begin{split}
-\frac{1}{2}e^{\frac{t^2}{2}}\sin p+\partial^2_{2}\tilde\Omega(t^2)&= \frac{ (\partial_2\lambda)^2 -c_{22}^{\gamma}\partial_{\gamma}\lambda}{ \partial_{p}\lambda},\\
&=\frac{ e^{t^2}\cos^2p -e^{t^2}}{2e^{\frac{t^2}{2}}\sin p }   \\
&=  -\frac{1}{2}e^{\frac{t^2}{2}}\sin p . \\
\end{split}
\end{equation}
Therefore, $\tilde\Omega(t^2)=0$ up linear terms.

To compute the integration constants of (\ref{Open Quantum cohomology 2}) is enough to find  $\tilde\Omega(t^2,t^3) $ such that 

 \begin{equation*}
\partial_{\alpha}\partial_{\beta}\left(\frac{e^{2p}}{2}+\sqrt{2}t^2e^p+t^1p-\frac{e^{t^3-p}}{\sqrt{2}}+\tilde\Omega(t^2,t^3)   \right)= \frac{ \partial_{\alpha}\lambda \partial_{\beta}\lambda -c_{\alpha\beta}^{\gamma}\partial_{\gamma}\lambda}{ \partial_{p}\lambda},
\end{equation*}
for $(\alpha,\beta)=(2,2),(2,3),(3,3)$.

%Then,
%\begin{equation}
%\begin{split}
%\frac{\partial \lambda}{\partial t^1}&=1,\\
%\frac{\partial \lambda}{\partial t^2}&=\sqrt{2}e^p,\\
%\frac{\partial \lambda}{\partial t^3}&=\frac{e^{t^3-p}}{\sqrt{2}},\\
%\frac{\partial \lambda}{\partial p}&=2e^{2p}+\sqrt{2}t^2e^p-\frac{e^{t^3-p}}{\sqrt{2}}\\
%\frac{\partial ^2\lambda}{\partial ^2t^2}&=0,\\
%\frac{\partial ^2\lambda}{\partial t^2\partial t^3}&=0,\\
%\frac{\partial ^2\lambda}{\partial ^2t^3}&=\frac{e^{t^3-p}}{\sqrt{2}},\\
%\end{split}
%\end{equation}
%Moreover,
%\begin{equation}
%\begin{split}
%c_{222}&=-t^2,\\
%c_{223}&=0,\\
%c_{233}&=e^{t^3},\\
%c_{333}&=t^2e^{t^3},\\
%\end{split}
%\end{equation}

Then,

 \begin{equation}
 \begin{split}
\partial_2\partial_2\tilde\Omega(t^2,t^3)&= \frac{ 2e^{2p} -c_{223}-c_{222}\frac{\partial \lambda}{\partial t^2}-\frac{\partial \lambda}{\partial t^3}}{ 2e^{2p}+\sqrt{2}t^2e^p-\frac{e^{t^3-p}}{\sqrt{2}} },\\
&= \frac{ 2e^{2p} +t^2\sqrt{2}e^p-\frac{e^{t^3-p}}{\sqrt{2}}}{ 2e^{2p}+\sqrt{2}t^2e^p-\frac{e^{t^3-p}}{\sqrt{2}} },\\
&=1,\\
\end{split}
\end{equation}

 \begin{equation}
 \begin{split}
\partial_2\partial_3\tilde\Omega(t^2,t^3)&= \frac{ e^{t^3} -c_{233}-c_{223}\frac{\partial \lambda}{\partial t^2}}{ 2e^{2p}+\sqrt{2}t^2e^p-\frac{e^{t^3-p}}{\sqrt{2}} },\\
&= \frac{ e^{t^3}-e^{t^3}}{ 2e^{2p}+\sqrt{2}t^2e^p-\frac{e^{t^3-p}}{\sqrt{2}} },\\
&=0.\\
\end{split}
\end{equation}

 \begin{equation}
 \begin{split}
- \frac{e^{t^3-p}}{\sqrt{2}} +\partial_2\partial_3\tilde\Omega(t^2,t^3)&=\frac{ \frac{e^{2t^3-2p}}{2} -c_{333}-c_{233}\frac{\partial \lambda}{\partial t^2}}{ 2e^{2p}+\sqrt{2}t^2e^p-\frac{e^{t^3-p}}{\sqrt{2}} },\\
&=\frac{ \frac{e^{2t^3-2p}}{2} -t^2e^{t^3}-\sqrt{2}e^{2t^3} }{ 2e^{2p}+\sqrt{2}t^2e^p-\frac{e^{t^3-p}}{\sqrt{2}} },\\
&=\frac{ \frac{e^{t^3-p}}{\sqrt{2}} \left[  \frac{e^{t^3-p}}{\sqrt{2}} -\sqrt{2}t^2e^{p}-2e^{2p}\right] }{ 2e^{2p}+\sqrt{2}t^2e^p-\frac{e^{t^3-p}}{\sqrt{2}} },\\
&=- \frac{e^{t^3-p}}{\sqrt{2}}.\\
\end{split}
\end{equation}
Therefore, $\tilde\Omega(t^2,t^3)=\frac{(t^2)^2}{2}$.

\subsection{$H_{1,1}$  case }

To compute the integration constants of (\ref{Open H11}) is enough to find  $\tilde\Omega(t^2,t^3) $ such that 

 \begin{equation}\label{integration constants of h11}
\partial_{\alpha}\partial_{\beta}\left(t^1p+(t^2)^2\left[\frac{\partial \log\theta_1(p,t^3)}{\partial p}\right]+\tilde\Omega(t^2,t^3)   \right)= \frac{ \partial_{\alpha}\lambda \partial_{\beta}\lambda -c_{\alpha\beta}^{\gamma}\partial_{\gamma}\lambda}{ \partial_{p}\lambda},
\end{equation}
for $(\alpha,\beta)=(2,2),(2,3),(3,3)$. In order to compute $\tilde\Omega(t^2,t^3) $, it is necessary to Laurent expand both sides of (\ref{integration constants of h11}) and check the constant in the variable $p$.\\

% Note that from the equation (\ref{part structure constant}), we have that 
% \begin{equation*}
%\partial_{\alpha}\partial_{\beta}\left(t^1p+(t^2)^2\left[\frac{\partial \log\theta_1(p,t^3)}{\partial p}\right]+\tilde\Omega(t^2,t^3)   \right)= \frac{ \partial_{\alpha}\lambda \partial_{\beta}\lambda}{ \partial_{p}\lambda} +O(p).
%\end{equation*}

The following relations will be useful

\begin{equation}\label{relation theta and Wp and Wz}
\begin{split}
\zeta(p,\tau)&=\frac{\partial \log \theta_1(p,\tau)}{\partial p}+4\pi ig_1(\tau)p,\\
\wp(p,\tau)&=-\frac{\partial^2 \log \theta_1(p,\tau)}{\partial^2 p}-4\pi ig_1(\tau),\\
\end{split}
\end{equation}
where $\zeta(p,\tau) ,\wp(p,\tau)$ are the Weierstrass zeta and Weierstrass p function respectively. i.e

\begin{equation*}
\begin{split}
\zeta(p,\tau)&=\frac{1}{p}+\sum_{m^2+n^2\neq 0}^{\infty}\frac{1}{p-m-n\tau}+\frac{1}{m+n\tau}+\frac{p}{(m+n\tau)^2},\\
\wp(p,\tau)&=\frac{1}{p^2}+\sum_{m^2+n^2\neq 0}^{\infty}\frac{1}{(p-m-n\tau)^2}-\frac{1}{(m+n\tau)^2}.
\end{split}
\end{equation*}
which have the following Laurent expansion around $p=0$
\begin{equation}\label{Laurent exp of Wp and Wz}
\begin{split}
\zeta(p,\tau)&=\frac{1}{p}+\sum_{n=2}^{\infty} \frac{c_n(\tau)}{2n-1}p^{2n-1}\\
\wp(p,\tau)&=\frac{1}{p}+\sum_{n=2}^{\infty} c_n(\tau)p^{2n}.
\end{split}
\end{equation}

For case $(2,2)$ in (\ref{integration constants of h11}), we have

 \begin{equation}\label{coeffiecient 22 of h11}
 \begin{split}
2\frac{\partial \log\theta_1(p,t^3)}{\partial p}+\partial_2\partial_2\tilde\Omega(t^2,t^3)&= \frac{ \partial_{2}\lambda \partial_{2}\lambda -c_{22}^{\gamma}\partial_{\gamma}\lambda}{ \partial_{p}\lambda},\\
&= \frac{ \partial_{2}\lambda \partial_{2}\lambda }{ \partial_{p}\lambda}+O(p).\\
\end{split}
\end{equation}

Substituting (\ref{Laurent exp of Wp and Wz}) and (\ref{relation theta and Wp and Wz}) in (\ref{coeffiecient 22 of h11}) and collecting the constant term in $p$ in the relation (\ref{coeffiecient 22 of h11}), we obtain
\begin{equation*}
\begin{split}
\partial_2\partial_2\tilde\Omega(t^2,t^3)=0.
\end{split}
\end{equation*}
Doing similar computation for the case $(2,3)$ and $(3,3)$, we obtain
\begin{equation*}
\begin{split}
\partial_2\partial_3\tilde\Omega(t^2,t^3)&=0,\\
\partial_3\partial_3\tilde\Omega(t^2,t^3)&=0,\\
\end{split}
\end{equation*}
which implies $\Omega(t^2,t^3)=0$ up linear terms.

%To compute the integration constants of (\ref{Open H100}) is enough to find  $\tilde\Omega(t^2,t^3,t^4) $ such that 
%
% \begin{equation}\label{integration constants of h100}
% \begin{split}
%\partial_{\alpha}\partial_{\beta}\left(t^1p+t^2\left[\log\left(\frac{\theta_1(p-t^3|t^4)}{\theta_1(p+t^3|t^4)}\right)\right]+\tilde\Omega(t^2,t^3,t^4)   \right)&= \frac{ \partial_{\alpha}\lambda \partial_{\beta}\lambda -c_{\alpha\beta}^{\gamma}\partial_{\gamma}\lambda}{ \partial_{p}\lambda},\\
%%&= \frac{ \partial_{\alpha}\lambda \partial_{\beta}\lambda }{ \partial_{p}\lambda}+O(p\pm t^3),\\
%\end{split}
%\end{equation}
%for $(\alpha,\beta)=(2,2),(2,3),(2,4),(3,3),(3,4),(4,4)$.

\end{document}